\documentclass[aps,prd,preprint,tightenlines,superscriptaddress,
amsfonts,amssymb,amsmath,showpacs,nofootinbib]{revtex4-2}

\usepackage{graphicx}%
\usepackage{multirow}%
\usepackage{amsmath,amssymb,amsfonts}%
\usepackage{amsthm}%

\usepackage{mathrsfs}%
\usepackage[title]{appendix}%

\usepackage{xcolor}%
\usepackage{textcomp}%
\usepackage{algorithm}%
\usepackage{algorithmicx}%
\usepackage{algpseudocode}%
\usepackage{colonequals}
\usepackage{listings}%

\newcommand{\bi}{\begin{itemize}}
\newcommand{\ei}{\end{itemize}}
\newcommand{\p}{\partial}

\def\be{\begin{equation}}
\def\ee{\end{equation}}
\def\bn{\begin{enumerate}}
\def\en{\end{enumerate}}
\newcommand{\nn}{\nonumber}
\newcommand{\bea}{\begin{eqnarray}}
\newcommand{\eea}{\end{eqnarray}}
\newcommand{\ba}{\begin{array}}
\newcommand{\ea}{\end{array}}
\newcommand{\bl}{\begin{align}}
\newcommand{\el}{\end{align}}

\newcommand{\al}{\alpha}

\newcommand{\ta}{\tau}

\newcommand{\g}{\gamma}

\newcommand{\lam}{\lambda}

\newcommand{\ep}{\varepsilon}
\newcommand{\eps}{\epsilon}
\newcommand{\bs}{\begin{subequations}}
\newcommand{\es}{\end{subequations} \noindent}
\newcommand{\half}{\tfrac{1}{2}}

\DeclareMathOperator{\sech}{sech}

\theoremstyle{thmstyleone}%
\newtheorem{theorem}{Theorem}
\newtheorem*{lemma}{Lemma}
\newtheorem{proposition}[theorem]{Proposition}%

\theoremstyle{thmstyletwo}%
\newtheorem*{remark}{Remark}%
\newtheorem*{corollary}{Corollary}

\theoremstyle{thmstylethree}%

\begin{document}

\title[Some exact results on the BKL scenario]
{Some exact results on the Belinski-Khalatnikov-Lifshitz scenario}


\author*{Piotr P. Goldstein}\footnote{email: \texttt{piotr.goldstein@ncbj.gov.pl}}
ORCID 0000-0002-0236-5332



\affiliation{Theoretical Physics Division National Centre for
Nuclear Research, Pasteura 7, Warsaw, Poland}


\date{\today}
\begin{abstract}
The well-known Belinski-Khalatnikov-Lifshitz (BKL) scenario for
the universe near the cosmological singularity is supplemented
with a few exact results following from the BKL asymptotic of the
Einstein equations: (1) The cosmological singularity is proved to
be an inevitable beginning or end of the universe as described by
these equations. (2) Attaining the singularity from shrinking
initial conditions requires infinite time parameter $\tau$; no
singularity of any kind may occur in a finite $\tau$. (3) The
previously found exact solution [P.G. and W. Piechocki, Eur. Phys.
J. C 82:216 (2022)] is the only asymptotic with well-defined
proportions between the directional scale factors which have been
appropriately compensated against indefinite growth of anisotropy.
In all other cases, the universe undergoes oscillations of Kasner
type, which reduce the length scales to nearly zero in some
directions, while largely extending it in the others. Together
with instability of the exact solution [op. cit.], it makes the
approach to the singularity inevitably chaotic. (4) Reduced
equations are proposed and explicitly solved to describe these
oscillations near their turning points. In logarithmic variables,
the oscillations are found to have sawtooth shapes. A by-product
is a quadric of kinetic energy, a simple geometric tool for all
this analysis.
\end{abstract}
\maketitle
\tableofcontents

\section{Introduction}\label{sec1}
The Belinski-Khalatnikov-Lifshitz (BKL) scenario plays a
significant role among classical descriptions of the universe
close to the cosmological singularity. It has been developed as an
answer to the question whether the singularity follows from the
Einstein equations for a set of nonzero measure on the manifold of
initial or final conditions \cite{LK}. The manifold should have
the proper dimensionality, i.e., depend on the proper number of
arbitrary parameters. Moreover, the singularity should be physical
reality rather than a result of simplifying assumptions of a
simplistic model or a choice of a special frame of reference.

The first question concerned the sheer existence of such a
singularity. The initial approach in \cite{LK} gave the negative
answer to this question. Namely, the authors demonstrated
possibility of transformation to the synchronous frame of
reference (in which time is the proper time at each point) and
proved that all solutions having the required properties, such
that the corresponding determinant of the spatial metric tensor
vanishes in a finite time, are fictitious because the singularity
may vanish in other reference systems. However, later R. Penrose
\cite{Penrose} and S. Hawking \cite{Hawking} proved existence of a
singularity independent of the frame of reference, first for the
case of a collapsing star \cite{Penrose}, then for a class of
cosmological models \cite{Hawking}. This made the authors of
\cite{LK} reconsider their assertion. To obtain a universe with a
physical singularity, BKL considered \cite{BKL} possible
generalizations of Kasner's homogeneous universe without matter
\cite{MTW}.

The original Kasner's solutions exactly satisfy the Einstein
equations without the energy-momentum term. They describe an
Euclidean metric; for orthogonal principal directions, it reads
\cite{MTW}
\be\label{Kasner}
ds^2=dt^2-(t^{2p_1}dx_1^2+t^{2p_2}dx_2^2+t^{2p_3}dx_3^2)
\ee
where the constants $p_1,\,p_2,\,p_3$ satisfy
\be\label{p-cond1}
p_1+p_2+p_3=1 ~\text{ and}\quad p_1^2+p_2^2+p_3^2=1.
\ee
As \eqref{p-cond1} infers that one of these exponents has to be
negative, the Kasner solutions are singular at $t=0$. Moreover,
they are anisotropic, and their anisotropy, measured by ratios of
scales in the principal spatial directions, grows indefinitely on
the approach to the singularity (with the exception of
$p_1=p_2=0,\,p_3=1$, up to exchange or the indices). The
differential of the volume is proportional to $t$, hence its limit
is equal to $0$ if $t\to 0^+$ .

The Kasner metric \eqref{Kasner} may be one of the possible
answers to another important question: on symmetries of the
primordial universe, close to the singularity. Although the recent
universe is isotropic, it does not mean that it has been isotropic
from the beginning. Therefore, cosmological models allowing for
primordial anisotropy should be taken into consideration.

The authors of \cite{BKL} looked for a generalization of the
Kasner solutions to a possibly large class preserving Kasner's
properties: homogeneity and increasing anisotropy on the approach
to the singularity. The asymptotic behavior of the universe under
these assumptions is known as the BKL scenario. Limiting the
interest to the region close to the singularity allows for
reduction of the Einstein equations to relatively simple ordinary
differential equations (ODE).

For completeness, we provide a short summary how these ODE follow
from their Einstein origin, as derived in \cite{Ryan}.

(i) In the neighborhood of the singularity, the influence of
matter, described by the r.h.s. of the Einstein equations, i.e.,
the energy-momentum tensor, may be neglected. This assumption is
based on the observation \cite{BKL} that the energy-momentum term
has higher order in time (if the singularity is at $t=0$),
compared to the singular behavior of the spacetime curvature. (ii)
For further calculations, the authors of \cite{BKL} choose the
synchronous frame of reference. This choice allows for assuming
the metric tensor in the form
\be\label{metric}
ds^2=dt^2-\g_{ab}(t)dx^a dx^b,
\ee
where $\g_{ab}$ are components of the spatial metric tensor in the
coordinate system $x^a,~a=1,2,3$ (summation convention is used to
covariant-contravariant pairs of indices).

Metric \eqref{metric}, substituted to the Einstein equations,
reduces the 4-dimensional (4D) problem to a problem of finding a
3D metric. In the Bianchi classification, the corresponding Lie
algebra is Bianchi~IX. (iii) Its structure constants are chosen as
$\ep_{abc}$, where $\ep$ is the Levi-Civita antisymmetric symbol.

(iv) For the sake of simplicity, $dt$ in \eqref{metric} is
rescaled by the factor of the spatial volume, according to
\be
dt=\sqrt{\g} d\tau,
\ee
where $\g$ is the determinant of the spatial (time-dependent)
metric tensor $\g_{ab}$ ($a,\, b =1,\, 2,\, 3$). In Kasner's
metric, we have $\sqrt{\g}=t$, whence $\tau$ is a logarithmic time
parameter, ${d\tau=d\ln t}$. This way, the singularity at $t=0$
appears as the limit $\tau\to -\infty$ or (by symmetry of the
resulting equations) $\tau\to\infty$. The latter is our choice of
this time parameter. With this choice, a collapse at
$\tau\to\infty$ will correspond to $t\to 0^+$ i.e. to the backward
trip down to the initial singularity.

(v) Of the ten Einstein equations for the spacetime without
matter, those describing the time-spatial, i.e. $^0_a$ components,
are found to provide only relations between constants; they do not
describe the dynamics. What remains are 3 diagonal and 3
off-diagonal equations for the spatial, i.e. $^b_a$ components,
and one equation for the temporal $^0_0$ component. (vi) The order
2 of the off-diagonal spatial equations is reduced to 1 with the
aid of the Bianchi identities. (vii) Two of three corresponding
constants of the integration are set to 0, which determines the
orientation of the chosen coordinate system (thus further
specifying the frame of reference).

(viii) In general, the principal axes of the spatial metric tensor
would rotate with respect to the chosen fixed frame. However, it
is proven that the rate of the rotation (e.g. in terms of the
Euler angles) tends to zero on the approach to the singularity,
provided that the ratios of the two shorter length scales along
the principal axes to the longest one tend to zero (this is the
indefinitely growing anisotropy). Hence, the angles tend to the
respective constant values on the approach to the singularity.
This reduces the task to solving one temporal ODE and three
spatial ODE for these length scales.

(ix) In the zero order in the aforementioned ratios of scales, we
obtain the required three spatial ODE, namely
\begin{equation}\label{L1}
\frac{d^2 \ln a  }{d \ta^2} = \frac{b}{a}- a^2,~~~~\frac{d^2 \ln b
}{d \ta^2} = a^2 - \frac{b}{a} + \frac{c}{b},~~~~\frac{d^2 \ln c
}{d \ta^2} = a^2 - \frac{c}{b}.
\end{equation}
where quantities $\,a=a(\ta),\, b=b(\ta)$ and $\,c=c(\ta)$ are, up
to constants of order $1$, proportional to squares of the length
scales in three principal directions of the chosen synchronous
reference system. They are known as directional scale factors. The
anisotropy assumption results in $a\gg b\gg c$.

These ODE are subject to the constraint imposed by the temporal
equation
\begin{equation}\label{L2}
\frac{d\ln a}{d\ta}\;\frac{d\ln b}{d\ta} + \frac{d\ln
a}{d\ta}\;\frac{d\ln c}{d\ta} + \frac{d\ln b}{d\ta}\;\frac{d\ln
c}{d\ta} = a^2 + \frac{b}{a} + \frac{c}{b} \, ,
\ee

Equations \eqref{L1}, together with \eqref{L2}, will be shortly
called the BKL equations. Due to their time-reversibility, they
may describe both, expansion of the universe starting from the
singularity or its final collapse.

Numerous papers, were devoted to both analytic and numeric
analyses of the BKL scenario, e.g.
\cite{BKL,BKL3,Ryan,Belinski:2014kba}. A Hamiltonian approach was
analyzed in detail in \cite{CP}, and a comparison with the
diagonal Mixmaster universe was done in \cite{CMX}. The scenario
was discussed in detail, on a broad background of related Bianchi
models, in the book by Belinski and Henneaux \cite{book}.

The goal of the present paper is to supplement their work by a few
exact results. All of them may be attained within the physics
described by the BKL equations \eqref{L1}, \eqref{L2}. The first
result was obtained in our previous paper with W. Piechocki
\cite{GP}, where we found an exact solution of these equations.
Namely, the solution reads
\be\label{solution}
a(\ta)= \frac{3}{\lvert \ta-\ta_0\rvert },~~ b(\ta)= \frac{30}{
\lvert \ta-\ta_0\rvert ^{3}},~~ c (\ta)= \frac{120}{\lvert
\ta-\ta_0\rvert ^{5}} \,
\ee
where   $\lvert \ta - \ta_0\rvert  \neq 0$ and $\ta_0$ is an
arbitrary real number.

In \cite{GP}, we also found that the exact solution
\eqref{solution} is unstable to small perturbations of the initial
conditions. In more detail, the perturbations evolve into two
oscillatory components. Although their amplitudes tend to zero as
$\ta\to\infty$, the instability is manifested in the growth of the
ratios of the perturbation amplitudes to the respective perturbed
scale factors; these ratios increase as $\ta^{1/2}$. A
characteristic value of the ratio between the two oscillation
frequencies (approximately equal to 2.06) is one of the results of
\cite{GP}; some chance exists that this ratio could have left
marks in the spectra of presently observed waves.

The new exact results which are demonstrated in this paper are
\bn
\item
All solutions of the BKL equations in which the initial volume
decreases with the time parameter $\ta$, i.e.
${dV/d\ta|_{\ta=0}<0}$, lead to the total collapse (in all three
directions) for $\ta\to\infty$ (subsection \ref{Sub5.2},
Proposition \ref{g_i=0}).
\item
For such initial conditions, no singularity occurs in finite
$\ta$, i.e., the scale factors remain finite and nonzero
(subsection \ref{cone}, Proposition \ref{regular}).
\item
The exact solution is the only one which collapses with
well-defined proportion between ratios of the directional scale
factors raised to the appropriate powers (the raising to power
compensates for the indefinitely growing anisotropy). For all
other solutions of \eqref{L1}, \eqref{L2}, the collapse is
attained through infinitely many oscillations between Kasner's
epochs (such behavior was already found in a different approach to
the asymptotics in \cite{BKL69} and in \cite{BKL}). If we approach
the singularity for $t\to 0^+$, the frequency of the oscillations
must tend to infinity and thus the approach is chaotic. Together
with the previously found \cite{GP} instability of this only
exception, it infers that chaos on the approach to the singularity
is inevitable (subsection \ref{Sub4.2}, Proposition
\ref{uni-apex}).
\item
BKL have shown that the dynamics is indeed Kasner-type and each
solution of this type loses its validity after some time
\cite{Ryan}. Then the role of the $p_i,~i=1, \,2,\,3$ in equations
\eqref{Kasner} and \eqref{p-cond1} is exchanged, i.e. transition
to another Kasner epoch occurs \cite{Ryan} (like in the Mixmaster
universe \cite{MTW}). This is the oscillatory approach to the
singular point, predicted in \cite{BKL}. In the present paper, we
provide explicit description of the dynamics in the transition
period between the epochs (subsection \ref{transition}). It
results (among other things) in finding that these are sawtooth
oscillations in the logarithmic variables, with explicitly
calculable gradients of the teeth.
\en
A by-product of the calculation is a cone (more general -- a
quadric) of kinetic energy, a simple geometric tool whose role is
similar to the well-known diagrams invented by Misner \cite{MTW};
it is described in detail in subsection \ref{cone}.

This paper is structured as follows:

In section \ref{Sec2}, the basic information on the BKL scenario
is shortly summarized. It includes the basic properties of
equations \eqref{L1}, \eqref{L2} and their Lagrangian-Hamiltonian
structure. Section \ref{methods} contains description of methods,
especially the geometric tool of the present analysis, which is
the cone of the kinetic part of the Lagrangian (further called
``kinetic energy''). Section \ref{Sec4} contains (in \ref{Sub4.2})
one of the main results, which is uniqueness of the exact solution
\eqref{solution} as the only one in which the collapse of the
universe has a definite proportion of the length scales (raised to
the appropriate powers to compensate for the indefinitely growing
anisotropy). In section \ref{Sec5}, the Kasner-type solutions are
described. A result stating that the BKL equations are not
satisfied by the exact Kasner solutions, even in the limit
$\ta\to\infty$, is given in subsection \ref{Sub5.2}. Section
\ref{quasi} is devoted to the Kasner-type solutions, with
subsection \ref{transition} discussing details of the transitions
between the adjacent Kasner epochs.

Most results are organized in a system of simple propositions and
their straightforward proofs. The proofs longer than a few lines
have been put off to two appendices.

\section{Basic properties of the BKL equations}\label{Sec2}
\indent\textit{Symmetries:} \cite{P-proc} The way in which
equations, \eqref{L1}, \eqref{L2} were obtained determines that
there is no symmetry under permutation of $a,\,b$ and $c$. On the
contrary, the growing anisotropy assumption results in $a\gg b\gg
c$. The system is evidently symmetric under time reversal
$\ta\rightleftarrows -\ta$; thus it can describe the universe in
both, a collapse or an explosion as its reversal. The equations
have two Lie symmetries \cite{P-proc}. The first one is a shift in
the time parameter $\ta\rightleftarrows \ta-\ta_0$ for any $\ta_0$
(which is obvious for an autonomous system). The second is a
scaling symmetry: If $a,\,b$ and $c$ constitute a solution of
\eqref{L1}, \eqref{L2}, and $\lam$ is the scaling parameter, such
that
\be\label{Lie}
\ta'=\lam \ta,~~a'=a/\lam,~~ b'=b/\lam^3,~~c'=c/\lam^5,
\ee
then $a',\,b'$ and $c'$ as functions of $\ta'$ make another
solution of the system \cite{P-proc}.

\textit{Dependence:} At first glance, system \eqref{L1},
\eqref{L2} appears to be overdetermined, due to the constraint
\eqref{L2} imposed on solutions of \eqref{L1}. However, the
constraint \eqref{L2} specifies a value of a first integral of
\eqref{L1}. Therefore, each of the equations \eqref{L1} may be
obtained from a system consisting of the other two of \eqref{L1}
and the constraint \eqref{L2}. E.g. \cite{P-proc}, if we
substitute $\ddot{a}$ and $\ddot{b}$ from the first two of the
equations \eqref{L1} into the $\ta$-derivative of the constraint
\eqref{L2}, we obtain the third equation of \eqref{L1} multiplied
by $(\dot{a}/a + \dot{b}/b)$ (the dot denotes $\ta$
differentiation). This way, the third equation of \eqref{L1} is
shown to be dependent on the other two and the constraint
\eqref{L2} (the expression in the parentheses yields $ab= const.$,
which is inconsistent with equations \eqref{L1}, \eqref{L2}).

\textit{Canonical structure} \cite{CP} Substitution
\be\label{a2x}
a=\exp(x_1),\quad b=\exp(x_2),\quad c=\exp(x_3)
\ee
yields a system derivable from a Lagrangian
\be\label{Lagrx}
\mathcal{L}=\dot{x}_1\dot{x}_2+\dot{x}_2\dot{x}_3+\dot{x}_3\dot{x}_1+\exp(2x_1)+\exp(x_2-x_1)+\exp(x_3-x_2),
\ee
with the constraint \eqref{L2} turning into \cite{CP}
\be\label{consx}
\mathcal{H}\colonequals\sum_{i=1}^3\frac{\p\mathcal{L}}{\p
\dot{x}_i}\dot{x}_i-
\mathcal{L}=\dot{x}_1\dot{x}_2+\dot{x}_2\dot{x}_3+\dot{x}_3\dot{x}_1-\exp(2x_1)-\exp(x_2-x_1)-\exp(x_3-x_2)=0.
\ee
Equation \eqref{consx} clarifies the sense of the dependence
between \eqref{L1} and \eqref{L2}: the constraint \eqref{L2} is a
particular choice of the first integral $\mathcal{H}$ for
solutions of equations \eqref{L1}, namely $\mathcal{H}=0$.

The Lagrangian \eqref{Lagrx} has a well defined potential and
kinetic ``energies''. The latter is an indefinite quadratic form
of signature $(+,-,-)$, whose zero surface is a cone. As seen from
\eqref{consx}, the ``potential energy'' is always negative while
the ``total energy'' is zero. This means that the ``kinetic
energy'' is positive, i.e., position of the system in the space of
``velocities'' varies inside a cone (further on, the quotation
marks will be omitted, also for the accelerations, i.e.,
derivatives of the velocities, as well as the kinetic, potential
and total energies).

\section{Methods}\label{methods}
We apply the aforementioned Lagrangian formalism, and illustrate
the evolution of the system by its trajectory in the space of
velocities in the diagonalized version of Lagrangian \eqref{Lagrx}
(defined below, in the first subsection).

\subsection{Useful variables}
Transformation \eqref{a2x} naturally replaces the original
variables $a,\,b,\,c$ by their logarithms $x_1,\,x_2,\,x_3$,
suitable for the Lagrangian description. However, the description
becomes clearer if we diagonalize the kinetic energy. If we care
about simplicity of the equations rather than unitarity of the
diagonalizing transformation (accepting its determinant to be
$-6$), a convenient substitution is
\be\label{x2u}
x_1=u_1-u_2-u_3,\quad x_2=u_1+2u_3,\quad x_3=u_1+u_2-u_3 \, ,
\ee
which yields the Lagrangian in the form diagonal in the velocities
$\dot{u}_1,\,\dot{u}_2,\,\dot{u}_3$,
\be\label{Lagru}
\mathcal{L}=3\dot{u}_1^2-\dot{u}_2^2-3\dot{u}_3^2+\exp\big(2(u_1-u_2-u_3)\big)+\exp(u_2-3u_3)+\exp(u_2+3u_3).
\ee
Variables $u_1,\,u_2,\,u_3$ define the principal directions in the
velocity space. The dynamics in the new variables is determined by
the Lagrange equations
\bs\label{Lagrequ}
\begin{align}
&\ddot{u}_1=\frac13 e^{2(u_1-u_2-u_3)},\label{Lagrequa}\\
&\ddot{u}_2=e^{2(u_1-u_2-u_3)}-e^{u_2}\cosh(3u_3),\label{Lagrequb}\\
&\ddot{u}_3=\frac13
e^{2(u_1-u_2-u_3)}-e^{u_2}\sinh(3u_3)\label{Lagrequc}.\\
&\text{with the constraint}\nn\\
\mathcal{H}\colonequals
&\,3\dot{u}_1^2-\dot{u}_2^2-3\dot{u}_3^2-e^{2(u_1-u_2-u_3)}-2e^{u_2}\cosh(3u_3)=0.\label{Lagrequcons}
\end{align}
\es

In terms of the original variables, the new ones are
\be\label{u2abc}
u_1=\frac13\ln(abc),\quad u_2=\frac12\ln(c/a),\quad
u_3=\frac16\ln(b^2/ac).
\ee
 As we can see, $u_1$ is the logarithm of the
volume scale, up to a multiplicative constant. Hence, the
diagonalization automatically separates dynamics of the volume
from that of the shape, thus doing what Misner introduced in the
first stage of his transformation for the Mixmaster model
\cite{MTW}.

The velocities might simply be expressed in terms of the canonical
momenta $p_i=\p \mathcal{L}/\p \dot{u}_i,~ i=1,2,3$; then
$\mathcal{H}$ would become the Hamiltonian, whose kinetic part was
also a diagonal quadratic form in the momenta. However, the
momenta are equal to the velocities, up to a multiplicative
constant. Therefore, we do not introduce extra momentum-variables.

Variables $a, \,b, \, c$ are not suitable for numerical
simulations, especially for their graphic presentation, because of
the disproportion between their sizes $a\gg b\gg c$. This purpose
is better served by quantities of equal order of magnitude. The
shape of equations \eqref{L1} and \eqref{L2} suggest that these
could be
\be\label{a2qrs}
q\colonequals a^2,\quad r\colonequals b/a,\quad s\colonequals c/b,
\ee
while their logarithmic counterparts
\be\label{qrs2y}
y_1\colonequals\ln{q},\quad y_2\colonequals\ln{r},\quad
y_3\colonequals\ln{s},
\ee
allow for the corresponding Lagrangian description. A simple
manipulation of the original equations \eqref{L1} leads to those
satisfied by the new variables, which may be cast into a compact
form
\be\label{L1qrs}
\left(\ba{c}\ln q\\ \ln r\\
\ln
s\ea\right)^{\!\centerdot\centerdot}=M\cdot\left(\ba{c}q\\r\\s\ea\right),
\text{~~or~~}\left(\ba{c}\ddot{y}_1\\ \ddot{y}_2\\
\ddot{y}_3\ea\right)=M\cdot\left(\ba{c}e^{y_1}\\ e^{y_2}\\
e^{y_3}\ea\right)
\ee
with the constraint given by
\be\label{L2qrs}
\frac12 \left(\ba{ccc}\ln q & \ln r &\ln s
\ea\right)^{\centerdot}\,\cdot
M^{-1}\cdot\left(\ba{c}\ln q\\ \ln r\\
\ln s\ea\right)^{\!\centerdot}-q-r-s=0
\ee
where the constant matrix $M$ is given by
\be\label{M}
M=\left(\ba{rrr}-2&\,2& 0\\2&-2&1\\0&\,1&-2\ea\right),\text{ with
 } \det M = 2, ~~ M^{-1}=\left(\,\ba{rrr}\tfrac32&\,2& 1\\2&2&1\\1&\,1&0\ea\,\right).
\ee
The simplicity of equations \eqref{L1qrs} is due to the fact that
the r.h.s. of the BKL equations \eqref{L1} are linear combinations
of $q,\,r$ and $s$ from \eqref{a2qrs}. This makes them convenient
variables for description of the BKL scenario. Their version for
$y_i,~i=1,\,2,\,3$ may be derived from a simple Lagrangian
\be\label{Lagry}
\mathcal{L}= \frac12
\left(\ba{ccc}\dot{y}_1&\dot{y}_2&\dot{y}_3\ea\right)\cdot
M^{-1}\cdot\left(\ba{c}\dot{y}_1\\ \dot{y}_2\\
\dot{y}_3\ea\right)+e^{\,y_1}+e^{\,y_2}+e^{\,y_3}
\ee
The constraint again corresponds to $\mathcal{H}=0$, where
$\mathcal{H}$ differs from the Lagrangian \eqref{Lagry}, by the
opposite signs at the exponential functions. Explicitly
\be\label{consy}
\mathcal{H} =\frac34 \dot {y} _ 1^2 +
 2 \dot {y}_1 \dot {y} _ 2 + \dot {y} _ 2^2 + \dot {y} _ 2 \dot {y}_3 + \dot {y}_3\dot {y}_1
 -\left( e^{\,y_1}+e^{\,y_2}+e^{\,y_3}\right)=0.
\ee
Diagonalization of the kinetic energy in the Lagrangian
\eqref{Lagry}, is achieved by substitution of $y_1,\,y_2$ and
$y_3$ with their values in terms of $u_1,\,u_2$ and $u_3$
respectively
\be\label{y2u}
y_1=2 (u_1 - u_2 - u_3),\quad y_2 = u_2 + 3 u_3,\quad y_3 = u_2 -
3 u_3,
\ee
which leads back to Lagrangian \eqref{Lagru} and the constrained
Lagrange equations which stem from it \eqref{Lagrequ}.

\subsection{The cone of kinetic energy\label{cone}}
Our basic geometric tool for analysis and presentation of the
dynamics will be the quadrics of kinetic energy.
\be\label{quadrics}
E_k\colonequals 3 \dot{u}_1^2-\dot{u}_2^2-3\dot{u}_3^2=\eps\ge 0.
\ee
For $\eps>0$ these are two-sheet hyperboloids, becoming a cone for
$\eps=0$.

 Assume that the initial conditions describe a universe, whose
volume is decreasing. In the variables $\dot{u}_1,\,\dot{u}_2$ and
$\dot{u}_3$, we have
\begin{proposition}
The dynamics of the universe which shrinks with $\ta$ takes place
in the lower interior of the cone
\be\label{l-i}
3\dot{u}_1^2-\dot{u}_2^2-3\dot{u}_3^2>0,\quad \dot{u}_1<0.
\ee
\end{proposition}
\begin{proof}
The first inequality (\textit{interior}) follows from the
constraint \eqref{Lagrequcons}, from which
$3\dot{u}_1^2-\dot{u}_2^2-3\dot{u}_3^2$ is equal to a sum of
exponential functions and hence it is positive. The second
(\textit{lower}, i.e. $\dot{u}_1<0$), is equivalent to the
assumption that the volume scale is decreasing, by the first
equation of \eqref{u2abc}.
\end{proof}

\begin{figure}
\begin{center}
\vspace{-2.5cm}
\includegraphics[width=1.2\textwidth]{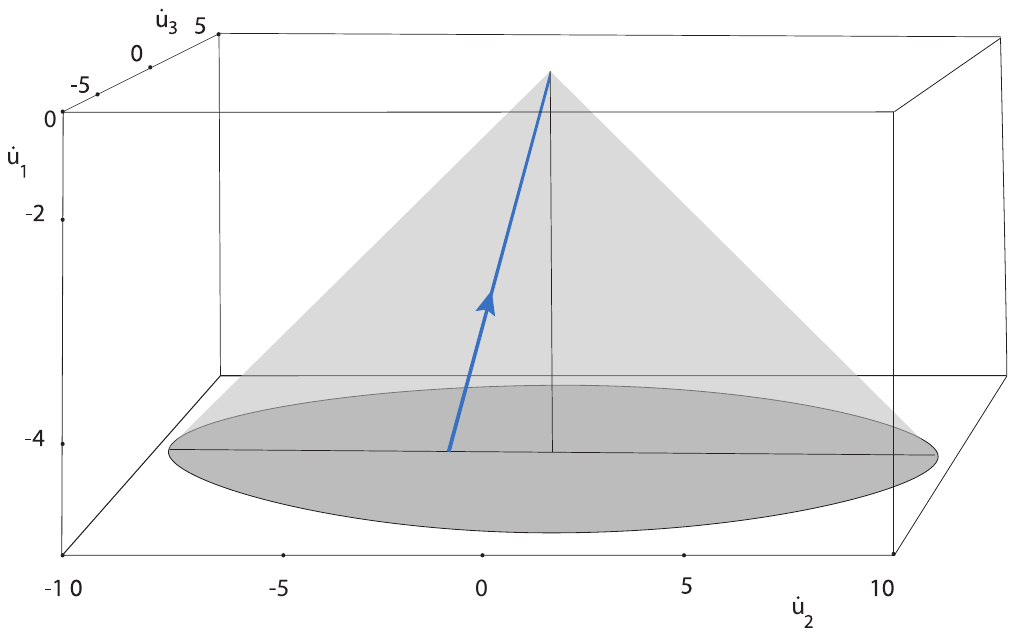}\vspace{-1.8cm}
\vspace{0.1cm} \caption{ (from \cite{GP}). The lower half (=
universe shrinking with $\ta$) of the cone
$3\dot{u}_1^2-\dot{u}_2^2-3\dot{u}_3^2
> 0$. The dynamics of the system takes place inside the cone.
The line with the arrow shows the exact solution; the arrow
indicates its direction of evolution. For $\ta\to\infty$, the line
tends to the apex of the cone. \\A position in the cone, together
with the tangent to the trajectory, provide complete information
on $u_1, u_2, u_3$, and their derivatives.\label{ccone1}}
\end{center}
\end{figure}

The following properties make the cone of kinetic energy a
particularly useful tool, reproducing essential information that
phase diagrams provide for single functions:
\begin{proposition}\label{singular}
Zero of the kinetic energy, i.e. the conical surface
$3\dot{u}_1^2-\dot{u}_2^2-3\dot{u}_3^2=0$ is a singular manifold
of the solutions.
\end{proposition}
\begin{proof}
From the constraint \eqref{Lagrequcons}, if the kinetic energy
$E_k$ turns to zero, then the sum of exponential functions (the
minus potential energy, $-E_p$) also has to be zero, whence all
exponents in \eqref{y2u} tend to $-\infty$ on approach to the
surface (including its apex). This requires that at least $u_1$
and $u_2$ tend to $-\infty$.
\end{proof}
\begin{corollary}
Bearing in mind that the accelerations $\ddot{u}_i,~i=1,\,2,\,3$,
are linear combinations of these exponential functions (see
\eqref{Lagrequ}), we see that all of them tend to $0$, i.e., the
rate of approach to the singular surface slows down to zero.
\end{corollary}
Later, we will see that the trajectory of the system gets
reflected from one of the hyperboloid surfaces \eqref{quadrics}
before attaining the singularity, with the exception of the exact
solution \eqref{sol-u}. With this exception, the conical surface
is never reached, even in the limit $\ta\to\infty$.

\begin{proposition}\label{complete}
The position of the system in the cone, together with the
direction of the tangent to the trajectory, provide complete
information on the local values of $u_1,\,u_2,\, u_3$ and their
$\ta$ derivatives.
\end{proposition}
\begin{proof}
The Cartesian coordinates of the position in the cone are the
components of the velocity, $\dot{u}_1,\,\dot{u}_2$ and
$\dot{u}_3$. The direction of the tangent yields proportions
between the components of the acceleration
$\ddot{u}_1,\,\ddot{u}_2$ and $\ddot{u}_3$. Given the components
of the velocity, the length of the acceleration vector can be
retrieved from
\be\label{length}
2\ddot{u}_2-9\ddot{u}_1+3\dot{u}_1^2-\dot{u}_2^2-3\dot{u}_3^2=0,
\ee
which is a simple linear combination of equations
\eqref{Lagrequa}, \eqref{Lagrequb} and \eqref{Lagrequcons} (the
exceptional case which would not yield the length is
$\ddot{u}_2/\ddot{u}_1=9/2$, but this might happen only on the
conical surface). Having the accelerations, we can calculate the
values of $u_1,\,u_2$ and $u_3$ by solving the system
\eqref{Lagrequa}, \eqref{Lagrequb},
\eqref{Lagrequc} for these variables.\\
By differentiation of these equations, we can obtain higher
derivatives of $u_i, ~~ i=1,2,3$, if they exist.
\end{proof}
\begin{remark}
Consequently, this means that the aforementioned data uniquely
determine the scale factors $a,\,b$ and $c$ and their derivatives
(by inversion of \eqref{u2abc}).
\end{remark}
\begin{proposition}\label{limits-dot}
Each of the velocities, $\dot{u}_1,\,\dot{u}_2$ and $\dot{u}_3$,
has a finite limit as $\ta\to\infty$.
\end{proposition}
\begin{proof}
A simple, reversible linear transformation turns the dynamics
equations \eqref{Lagrequa}, \eqref{Lagrequb}, \eqref{Lagrequc},
into
\bs\label{for-exp}
\begin{align}
3\ddot{u}_1=&e^{2(u_1-u_2-u_3)},\label{for-expa}\\
4\ddot{u}_1-\ddot{u}_2-\ddot{u}_3=&e^{u_2+3u_3}\\
2\ddot{u}_1-\ddot{u}_2+\ddot{u}_3=&e^{u_2-3u_3}
\end{align}
\es
The r.h.s. of the resulting equations are exponential functions
and hence they are positive, whence their l.h.s. are second
derivatives of convex functions and first derivatives, of
increasing functions, $\dot{u}_1,~4\dot{u}_1-\dot{u}_2-\dot{u}_3$
and $2\dot{u}_1-\dot{u}_2+\dot{u}_3$, respectively. The latter
functions are bounded, because the increasing property of
$\dot{u}_1$, together with $\dot{u}_1(0)<0$, infer
$\lvert\dot{u}_1(t)\rvert<\lvert\dot{u}_1(0)\rvert$, while both
$\lvert\dot{u}_2(t)\rvert$ and $\lvert\dot{u}_3(t)\rvert$ are not
greater than $\sqrt{3}\lvert\dot{u}_1(t)\rvert$ as long as we are
inside the cone. Hence, all three linear combinations of the first
derivatives are increasing functions bounded from above, and thus
have finite limits.

The limits $g_1\!\colonequals \lim_{\ta\to\infty}\dot{u}_1,~
g_2\!\colonequals \lim_{\ta\to\infty}\dot{u}_2$ and
$g_3\!\colonequals \lim_{\ta\to\infty}\dot{u}_3$ may be uniquely
regained from the limits of these linear combinations by inverting
the transformation which leads from \eqref{Lagrequ} to
\eqref{for-exp}.
\end{proof}
\begin{corollary}
As $\dot{y}_i,~~i=1,2,3$, are linear combinations of $\dot{u}_i$
\eqref{y2u}, they also have finite limits. Moreover, as the
corresponding linear transformation is nonsingular, all
$\dot{y}_i$ vanish at the apex and only at the apex. Also, for all
$i$, $y_i\to -\infty$ for $\ta\to\infty$, which is a consequence
of the constraint \eqref{consy}.
\end{corollary}

\begin{proposition}\label{infinite}
A trajectory which ends on the surface or at the apex of the cone,
needs infinite time $\ta$ to reach it.
\end{proposition}
\begin{proof}
Consider a trajectory beginning in the lower half of the cone and
ending on its surface or apex. Since, $\dot{u}_1<0$, hence $u_1$
is a decreasing function of $\ta$. On this basis, $\ta$ may be
calculated as
\be\label{time1}
\ta=\int_{u_1(0)}^{u_1}du_1'/\dot{u}_1'
\ee
 From \eqref{Lagrequa} and \eqref{l-i}, we have
$0>\dot{u}_1(t)\ge\dot{u}_1(0)$ in the lower interior of the cone,
whence $1/\dot{u}_1(t)\le\dot{u}_1(0)<0$. Hence the integrand
$1/\dot{u}_1'$ is separated from $0$ in the interval of
integration. On the other hand, $u_1\to -\infty$ when we approach
the boundary (see the proof of Proposition \ref{singular}). The
integral \eqref{time1} in the limit $u_1\to -\infty$ extends over
infinite interval, while its integrand is separated from zero.
Hence, it is infinite.
\end{proof}
\begin{remark}
The time parameter $\ta$ calculated in \eqref{time1}, over a
finite or infinite interval, is always positive, because the
integrand is negative, while the lower limit of integration is
greater than the upper limit.
\end{remark}

\begin{proposition}\label{endoftraj} $\mathrm{(converse~of~Proposition~
\ref{infinite})}$
\\ In the limit $\ta\to \infty$, each trajectory
reaches the surface or apex of the cone.
\end{proposition}
\begin{proof}
Time $\ta$ can also be expressed as
\be\label{time2}
\ta=\int_{\dot{u}_1(0)}^{\dot{u}_1}d\dot{u}_1'/\ddot{u}_1',
\ee
because $\dot{u}_1$ is an increasing function of $\ta$ due to the
positive sign of its derivative \eqref{Lagrequa}. For any point of
the lower interior of the cone, the denominator is greater than
zero (from \eqref{Lagrequa}), whence the integrand is finite and
so are the limits of integration. Hence $\ta$ has a finite value.
Merely for $\left(\dot{u}_1,\,\dot{u}_2,\,\dot{u}_3\right)$ lying
on the surface or at the apex of the cone can $\ddot{u}_1$ vanish
and $\ta$ become infinite.
\end{proof}
\begin{corollary}
This means that points on a trajectory in the interior of the cone
are attained in a finite time.
\end{corollary}
\begin{remark}
Variable $\dot{u}_1$ (which is the logarithmic derivative of the
volume scale) might in principle replace time parameter $\ta$, as
it is an increasing function of $\ta$. Nevertheless, its use for
this purpose is limited, as its variation is very uneven (see Fig.
\ref{ui}). There are intervals of $\ta$ where the exponential
function in \eqref{Lagrequa} is close to zero and thus $\dot{u}_1$
hardly increases. This happens at each approach to the conical
surface, (i.e. when the kinetic energy decreases to nearly $0$).
Then also the other exponential components in \eqref{Lagrequ}
become very small (see Fig. 2, 3). We discuss this Kasner-type
behavior in subsection \ref{transition}.
\end{remark}

\begin{proposition}\label{regular}
Trajectories in the interior of the cone contain nonsingular,
nonzero points of solutions to \eqref{L1}, \eqref{L2}.
\end{proposition}
\begin{proof}
Consider points in the interior of the cone, lying on a trajectory
of a solution in the $\dot{u}_i$ space, $~i=1,2,3$. The
coordinates $u_i$ are finite as integrals of finite $\dot{u}_i$
over a finite $\ta$ interval. Hence $y_i$ are also finite as
linear combinations of the respective $u_i$, whence $a^2,\;b/a$
and $c/a$ are finite and positive as they are equal to $\exp
y_1,\;\exp y_2$ and $\exp y_3$ respectively (see \eqref{qrs2y}).
This guarantees that also the scale factors $a,\,b$ and $c$ are
finite and nonzero.
\end{proof}

\noindent Finally,
\begin{proposition}
There is no possibility of a stop in the interior of the cone
(i.e., each point of positive kinetic energy corresponds to
nonzero acceleration).
\end{proposition}
\begin{proof}
This property follows directly from equation \eqref{length}. As
long as $3\dot{u}_1^2-\dot{u}_2^2-3\dot{u}_3^2>0$, we have
$2\ddot{u}_2-9\ddot{u}_1<0$, which requires at least one nonzero
component of the acceleration.
\end{proof}

Situations where the trajectory in the velocity space slows down
to almost full stop may happen (see Fig. 2), which corresponds to
the Kasner-type behavior discussed in subsection \ref{transition}.
\section{Solutions ending at the apex of the cone}\label{Sec4}
\subsection{The exact solution in the cone}
In terms of the $u_i$ variables, the exact solution reads
\be\label{sol-u}
u_1=\frac13\ln\frac{10800}{\lvert \ta-\ta_0\rvert^9} ,\quad
u_2=\half \ln\frac{40}{\lvert \ta-\ta_0\rvert^4},\quad
u_3=\frac16\ln\frac52.
\ee
Obviously, for $\ta>\ta_0$, both $\dot{u}_1$ and $\dot{u}_2$ are
proportional to $(\ta-\ta_0)^{-1}$, while $\dot{u}_3=0$. We also
have $\dot{u}_2=\frac23\dot{u}_1$, which means that the trajectory
corresponding to the exact solution is a half-line whose end lies
at the apex of the cone (see Fig. 1). Physically, it describes a
power-like collapse of all scale factors $a,\, b,\, c$ to zero,
i.e. a collapse of the universe, in all directions, to a point, as
$\ta\to\infty$ (which corresponds to the original time tending to
zero from the right).

The instability of the exact solution, found in \cite{GP} affects
also the solution in terms of $u_i$, only the coefficients at the
oscillatory terms are different. However, the solution itself is
regular up to the apex.

\subsection{On the uniqueness of the exact solution}
\label{Sub4.2}
 Due to the assumed indefinitely growing anisotropy, the ratios of the
scale factors, $a,\,b,\,c$, always tend to $0$ or $\infty$ for
$\ta\to\infty$. Therefore, the limits of their proportions at
$\ta\to\infty$ do not distinguish between various solutions. If we
want to see the differences between their asymptotics, we have to
raise them to appropriate powers to compensate for the anisotropy.
E.g., if $a$ is left at power $1$, we have to raise $b^{1/3}$ and
$c^{1/5}$ according to the Lie symmetry \eqref{Lie} (this can also
be seen in the exact solution \eqref{solution}, where
$b/a^3=10/9$, while $c/a^5=40/81$).

Hence the question from the title of the subsection may be
formulated as: ``is the all-direction collapse possible with any
other proportion of the scale factors at appropriate powers?''.

Instead of proportion between these powers of $b$, $c$ and $a$, we
may use their $1:1$ counterparts in the proportions between $q, \,
r,\,s$ (defined in \eqref{a2qrs}), without this somewhat
artificial raising to powers. This is possible thanks to
\be\label{ratios}
b/a^3=r/q,\qquad c/a^5=\left(r/q\right)\left(s/q\right).
\ee
This way, if finite or infinite limits of the l.h.s.'s of
\eqref{ratios}, exist, then also the ratios $r/q$ and $s/q$ have
well defined limits, and vice versa. Therefore, we may consider
the limits of the latter. Further, if they exist, then we have
also well defined limits of their logarithms $y_2-y_1$ and
$y_3-y_1$ respectively. By linear transformation \eqref{y2u} and
the de l'H\^opital rule, we can also interpret this fact in our
conical picture. Namely, this means that trajectories exist which
approach the apex in a definite direction as $\ta\to\infty$.

The answer to the question of uniqueness is positive. Namely
\begin{proposition}\label{uni-apex}
The only asymptotic of solutions to equations \eqref{L1qrs},
\eqref{L2qrs}, which ensures existence of limits of
$lim_{\ta\to\infty}(r/q)$ and $lim_{\ta\to\infty}(s/q)$ is that of
the exact solution, i.e.
\be\label{solqrs}
q= \frac{9}{(\ta-\ta_0)^2},\quad r= \frac{10}{(\ta-\ta_0)^2},\quad
s= \frac{4}{(\ta-\ta_0)^2},
\ee
which is equivalent to \eqref{solution} upon substitution
\eqref{a2qrs}.
\end{proposition}
The proof has been put off to Appendix A.

\begin{remark}
Existence of the above limits means (by de l'H\^opital's rule)
that also the corresponding limits $lim_{\ta\to\infty}
\dot{r}/\dot{q}$ and $lim_{\ta\to\infty}\dot{s}/\dot{q}$ exist,
because the constraint \eqref{L2qrs} requires that $q\to 0,\,r\to
0,\,s\to 0$ whenever the kinetic term vanishes, i.e. for any
$y_1,\,y_2,\,y_3$ on the conical surface or at the apex.
\end{remark}

\section{Solutions approaching the lateral surface of the cone}\label{Sec5}
From Proposition \ref{endoftraj}, we know that the trajectory has
to approach the apex (as the exact solution) or the lateral
surface of the cone. In this section, we discuss the latter case.

\subsection{Asymptotics of the diagonal velocities}
Let a trajectory approach the conical surface, but not the apex.
According to Proposition~\ref{limits-dot}, all three velocities
have their limits at $\ta\to\infty$, namely $\dot{u}_i\to
g_i,~~i=1,2,3,~~g_1\ne 0$. If the limits were on the surface, they
would satisfy the equation of the cone
\be\label{g-cone}
3g_1^2-g_2^2-3g_3^2=0
\ee
With $\dot{u}_i\to g_i$, the asymptotic behavior of the diagonal
variables is $u_i\sim g_i \ta$. Translating dependence between
$u_i$ and $a,\,b,\,c$ \eqref{u2abc} into the corresponding
asymptotics of the scale factors, according to \eqref{u2abc}, we
obtain
\be
a\sim\exp(p_1 \ta),\quad b\sim\exp(p_2 \ta),\quad c\sim\exp(p_3
\ta)
\ee
where $ p_1=g_1 - g_2 - g_3,~~ p_2= g_1 + 2 g_3,~~p_3= g_1 + g_2 -
g_3$.

We could put any common coefficient in front of $p_i,\,i=1,2,3$ by
changing the scale of the time parameter $\ta$ (this was done in
\cite{Ryan}). By straightforward calculation, equation
\eqref{g-cone} turns into a constraint on the constants $p_i$
\be\label{p-cond}
p_1p_2+p_2p_3+p_3p_1=0,
\ee
which by rescaling of $\ta$ (including its direction) so that
$p_1+p_2+p_3=1$ (first Kasner's condition \eqref{p-cond1}
\cite{MTW}) is equivalent to the second Kasner's condition
\eqref{p-cond1}, in accordance with \cite{Ryan}.  This result
means that if a trajectory of a solution ended on the conical
surface, it would describe an exact Kasner's solution: the
universe would be squeezed to zero in one direction while being
stretched to infinity in the remaining two. Conversely, if $p_i$
satisfy the Kasner conditions \eqref{p-cond1}, then the
corresponding point $(g_1,g_2,g_3)$ lies on the surface of the
cone.

\subsection{Impossibility of the exact Kasner
asymptotics\label{Sub5.2}} Numerical calculation (see Fig.
\ref{kin-en}) shows that the Kasner solutions are reproduced with
high precision by solutions of the BKL. However, the
aforementioned exact Kasner-type solutions, though predicted and
described in \cite{Ryan}, do not exactly satisfy the BKL equations
\eqref{L1}, \eqref{L2} (which is acceptable as these equations are
approximate). In terms of the cone of kinetic energy, even in the
limit $\ta\to\infty$ does not the trajectory of the system reach
the lateral conical surface $E_k=0,~\dot{u_1}<0$. Namely, it gets
reflected from a hyperboloid $E_k=\eps>0$ before reaching the
boundary of the cone, although the minimum $\eps$ of the kinetic
energy may be very small (see Fig. \ref{kin-en}). This property
means that
\begin{proposition}\label{g_i=0}
The only possible asymptotic behavior of solutions to
\eqref{L1qrs}, which satisfies constraint \eqref{consy},
corresponds to $\g_1=\g_2=\g_3=0$, where
$\g_i=lim_{\ta\to\infty}y_i,~~i=1,2,3$.
\end{proposition}
\begin{remark}
This means that all solutions eventually end at the apex, i.e.,
the fate of the universe is a total collapse to a point (with the
original time -- if we follow the history of the universe
backwards, we end at a point). According to Proposition
\ref{infinite}, this happens only in the limit $\ta\to\infty$.
Moreover, together with Proposition \ref{uni-apex}, it means that
the collapse is always chaotic.
\end{remark}
 The proof is lengthy and therefore it has been put
off to Appendix B.

\section{Quasi-Kasner solutions}\label{quasi}
\subsection{Description}
Although exact Kasner solutions do not satisfy the system
\eqref{L1},\eqref{L2}, numerical calculations show that
approximate Kasner-type solutions of these equations are possible
and precise. Namely, the trajectories may approach the surface of
the cone and bounce at a short distance from it, thus switching
the universe to what may be considered the next Kasner epoch. The
trajectory then passes through the interior of the cone until it
approaches another point almost on its surface, at a less negative
value of $\dot{u}_1$ (as this coordinate may only increase,
according to \eqref{for-expa}). As the cone narrows, the amplitude
of this quasi-periodic oscillations diminishes. This behavior
corresponds to reflections from the potential walls on Misner's
diagrams \cite{MTW}, while the surfaces of the corresponding
quadrics (lower halves of the two-sheet hyperboloids)
\be\label{hyper}
3 \dot{u}_1^2-\dot{u}_2^2-3\dot{u}_3^2=\eps_n
\ee
play the role of the equipotentials. The index $n$ numbers the
reflections while $\eps_n$ is a measure of its closeness to the
surface of the cone.

Since the volume scale is proportional to $\exp{\left(\tfrac32
u_1\right)}$, while the time derivative, $\dot{u}_1$, is negative
throughout the evolution, the universe becomes more compact at
subsequent reflections, although the reduction need not concern
the scales in all directions.

Apparently, as seen in Fig. \ref{ui}, the velocity components
$\dot{u}_i$ seem to remain constant for some time and the kinetic
energy looks as if it were equal to zero in Fig. \ref{kin-en}. A
logarithmic scale is necessary to reveal the actual variation of
these quantities in that figure. This is due to the exponential
dependence of the derivatives on the values of these variables.

\begin{figure}
\begin{center}
\includegraphics[width=0.8\textwidth]{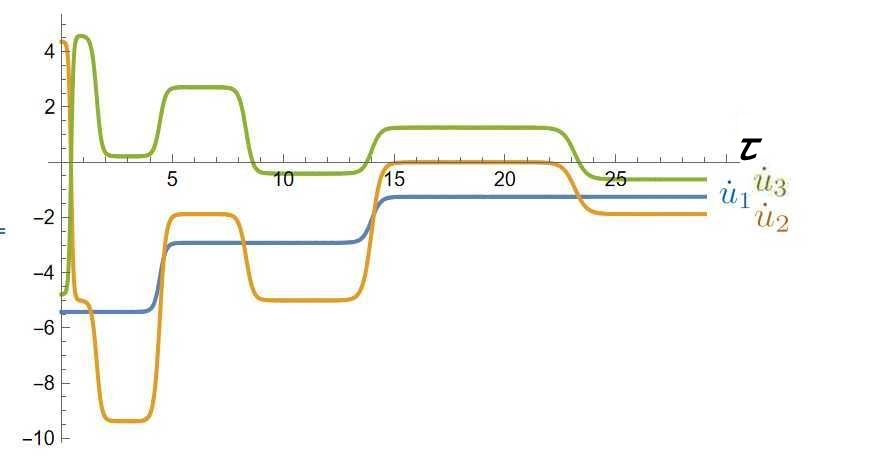}
\caption{Three components of the velocity, $\dot{u}_1,\,\dot{u}_2$
and $\dot{u}_3$, as functions of the time parameter $\ta$. The
initial conditions correspond to $q(0)=0.225,\,
r(0)=0.25,\,s(0)=0.1,\,\dot{q}(0)=-2,25,\,\dot{r}(0)=-2,5$. Each
of the velocity components has time intervals of apparently
constant value. Revealing their variability requires a logarithmic
scale, as seen in the next figures.\label{ui}}
\end{center}
\end{figure}

\begin{figure}
\begin{center}
\includegraphics[width=0.8\textwidth]{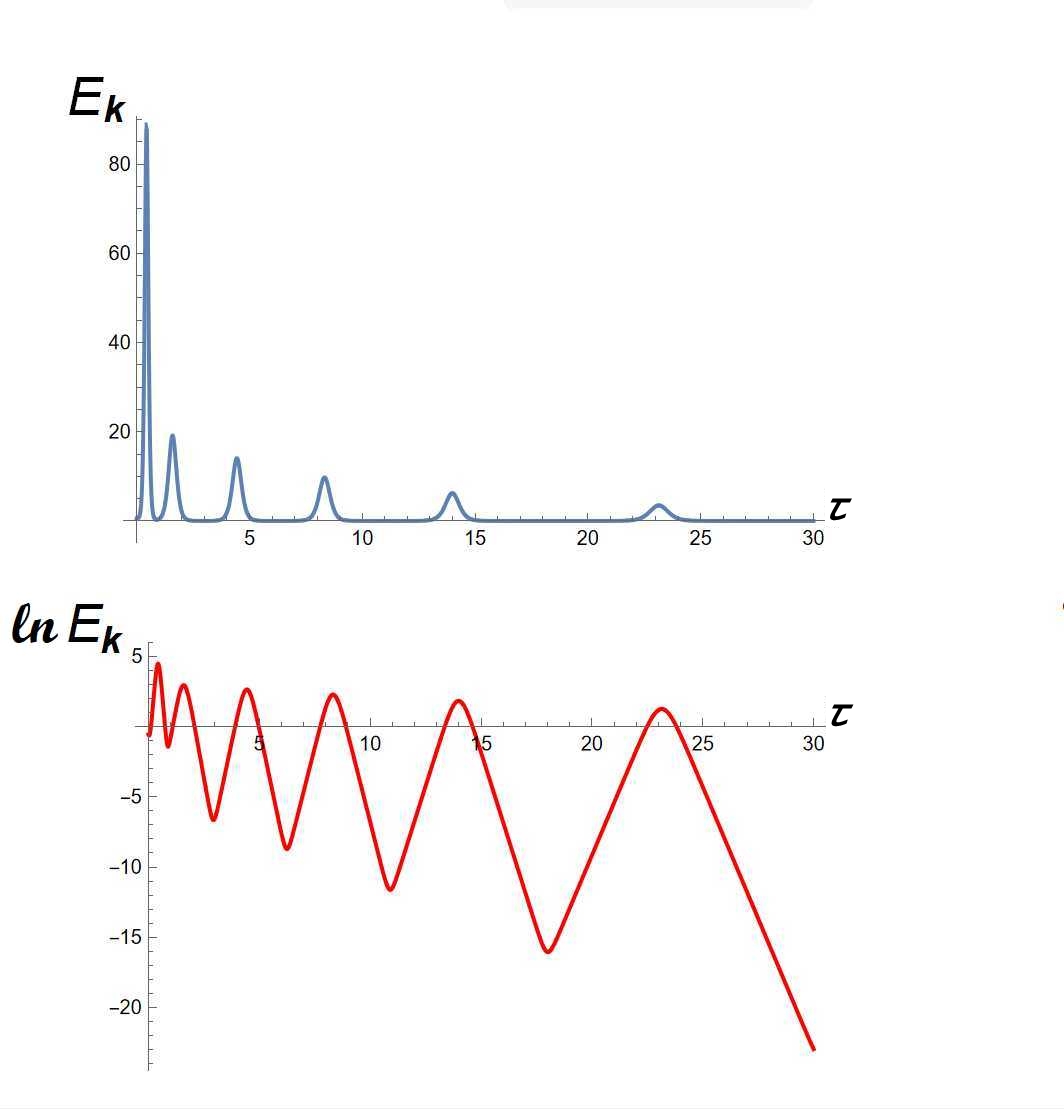}
\caption{The kinetic energy as a function of the time parameter
$\ta$ for $q(0)=0.225,\,
r(0)=0.25,\,s(0)=0.1,\,\dot{q}(0)=-2,25,\,\dot{r}(0)=-2,5$. In the
upper graph, apparently, $E_k$ systematically reaches zero
corresponding to the surface of the cone, and stays at this level
for a long time, but the logarithmic scale in the lower graph
reveals its sawtooth oscillations, with the teeth a little rounded
(of shape $\ln\, \sech^2$) at the reflections from the
hyperboloidal surfaces \eqref{hyper}.\label{kin-en}}
\end{center}
\end{figure}
\subsection{Between Kasner's epochs\label{transition}}
In \cite{Ryan}, the consecutive switching between Kasner's epochs
was extended from its original Mixmaster universe to the BKL
scenario. This was achieved by assuming the l.h.s.'s of equations
\eqref{L1} to be zero, which was justified as an approximation of
\eqref{L1}, \eqref{L2} for large $\ta$ and slow variation. The
solution was obtained in the form \cite{Ryan}
\be\label{Ryan-app}
a=A_a\exp(2\Lambda p_a t),\quad b=A_b\exp(2\Lambda p_b t),\quad
c=A_c\exp(2\Lambda p_c t),
\ee
where $p_a,\,p_b$ and $p_c$ satisfied the Kasner conditions
\eqref{p-cond1}. The authors noticed that the exponential
dependence of the scale factors caused growing of the r.h.s. in
equations \eqref{L1} so that the approximation was no more valid
after some time (or rather some $\ta$). They stated that this led
to the exchange of the role of two scale factors, i.e. transition
to the next Kasner epoch.

In this section, we examine this transition in more detail. The
best variables for this purpose are the exponents
$y_1,\,y_2,\,y_3$ and the corresponding exponential functions
$q,\,r,\,s$, satisfying equations \eqref{L1qrs} and \eqref{L2qrs}.

Let the trajectory described by the trio
$(\dot{y}_1,\,\dot{y}_2,\,\dot{y}_3)$ closely approach a point
$\g_1,\,\g_2,\,\g_3$ on the lateral surface of the cone. Without
touching it, the trajectory would ``bounce'' from a potential wall
according to \eqref{hyper} close to this point. Let $\ta=0$ be the
moment when the bounce occurs. Then the asymptotics of
$y_i,\;i=1,2,3$ would be
\be\label{asympt}
y_i^0+\g_i \ta,~~i=1,\,2,\,3,
\ee
where $y_i^0$ are constants. In the generic case, all $\g_i$ would
be different. Hence, as $\ta$ increases, one of the exponents
would soon become significantly larger than the two others; this
would make one of the exponential functions $q=\exp{y_1},\, r=\exp
y_2,\,s=\exp y_3$ much greater than the remaining two. This
disproportion of sizes allows for a reduction of \eqref{L1qrs} by
retaining the largest of them in \eqref{L1qrs} while approximating
the other two by $0$ in some neighborhood of the turning point.
The present approximation goes one step beyond the approximation
of \cite{Ryan}, which neglected all exponential functions, and
thus it yields more information on the transition period. Note
that the kinetic energy is equal to minus potential energy,
$E_k=q+r+s$ (according to \eqref{L2qrs}). In the above
approximation it reduces to only one term: the greatest of these
three variables.

Whichever of the $y_i$ is the greatest, we obtain an easily
solvable system of ODE. If the dominant exponent is $y_1$ (for
example), then the solution, with physically acceptable signs of
the integration constants, after a short time from the turning,
reads
\bs\label{sol11}
\begin{align}
q\colonequals &\,a^2=k_q^2 \sech^2 k_q (\ta - \ta_+),\label{sol11q}\\
r\colonequals &\,b/a=k_r^2 e^{\beta \ta} \cosh^2 k_q (\ta - \ta_+),\label{sol11r}\\
s\colonequals &\,c/b=k_s^2 e^{\g \ta},\label{sol11s}
\end{align}
\es
where $\ta_+,\, k_q,\,k_r,\,k_s,\,\beta$ and $\g$ are real
constants, of which $\ta_+>0,$ while $k_q,\,k_r$ and $k_s$ may be
assumed positive. Consistency requires that $q\gg r$ and $q\gg s$
in some neighborhood of $\ta=\ta_+$. Hence $k_q\gg k_r$ and $k_q
\gg k_s$. This ordering complies with the neglect of $r=\exp{y_2}$
and $s=\exp(y_3)$ in \eqref{sol11}.

The constraint \eqref{L2qrs} limits the possible constant
parameters; it requires $k_q^2=\beta(\beta+\g)$, without any
condition imposed on the other parameters.

At the turning point, an exchange of the dominant exponents should
take place in the Kasner-type solutions. Indeed on the other side
of the turning point the dominant term should be $y_2$ or $y_3$,
as $y_1$ decreases with the distance from $\ta_+$ (and so does
$q$), while at least one of $y_2$ (and $r$) or $y_3$ (and $s$)
increases. It depends on the values of $\beta$ and $\g$ which of
these would overcome $y_1$ (and $q$) at some distance from
$\ta_+$. Consider, e.g. a transition from the situation where
$y_2$ is the greatest of $y_i$ for a little negative $\ta$; after
the bounce at $\ta=0$, $y_1$ becomes the greatest. Then, in some
interval to the left of $\ta=0$, we would have $r\gg q$ and $r\gg
s$. The solution for these a little negative $\ta$ would read
\bs\label{sol21}
\begin{align}
q&={k'_q}^2 e^{\al \ta}\cosh^2 k'_q (\ta - \ta_-),\label{sol21q}\\
r&={k'_r}^2 \sech^2 k'_q (\ta - \ta_-),\label{sol21r}\\
s&={k'_s}^2 e^{\g_1 \ta}\cosh k'_q (\ta - \ta_-),\label{sol21s},
\end{align}
\es
where the integration constants in \eqref{sol21}, $\al, \ta_-<0$,
as well as positive $k'_q,\,k'_r,\,k'_s,$,  correspond to their
counterparts for $\ta>0$ in \eqref{sol11}.
Solutions \eqref{sol11} with \eqref{sol21} provide detailed
picture of the dynamics in the right and left neighborhoods of the
turning point $\ta=0$. With the $\ta$-distance from the turning
point, the solutions containing the $\cosh$ and $\sech$ functions
very soon turn into their asymptotics due to the rapid decay of
their $\exp$ components having negative exponent, compared to
those with the positive one. As a consequence, the asymptotics of
their logarithms $y_1=\ln q,~y_2=\ln r$ and $y_3=\ln s$ are
piecewise linear functions, e.g.
\be\label{y1as}
y_1=\ln (4k_q^2)-2k_q\lvert\ta-\ta_+\rvert,\quad y_2=\ln
(k_r^2/4)+\beta\ta+2 k_q\lvert\ta-\ta_+\rvert
\ee
on the positive side of $\ta=0$, and
\be\label{y2as}
y_1=\ln ({k'_q}^2/4)+\al\ta+2 k'_r\lvert\ta-\ta_-\rvert, \quad
y_2=\ln (4 {k'_r}^2)-2 k'_r\lvert\ta-\ta_-\rvert
\ee
on the negative side. This way, each of these logarithmic
variables undergoes a sawtooth oscillation (the tooth is blunt as
the actual functions are logarithms of hyperbolic $\sech$ and
$\cosh$ rather than the absolute values of $\ta-\ta_+$ and
$\ta-\ta_-$). The linearity in the neighborhood of $\ta=0$ is
consistent with the assumed equations \eqref{asympt}, provided
that we adjust the parameters of the lines, namely, they must
satisfy
\begin{align}
&k_q=\frac{\g_1}{2},~k'_r=-\frac{\g_2}{2},~,k'_q=\frac{2}{\lvert\g_1\rvert}\exp\frac{y_1^0+y_2^0}{2},~k_r=\frac2{\lvert\g_2\rvert}\exp\frac{y_1^0+y_2^0}{2}\nn\\
&\al=\beta=\g_1+\g_2,~~\ta_+=\left(\ln\g_1^2-y_1^0\right)/\g_1,~~\ta_-=\left(\ln\g_2^2-y_2^0\right)/\g_2.
\end{align}
This way, the position $(\g_1,\,\g_2)$ of the bounce, together
with the values $y_1^0$ and $y_2^0$ completely determine the local
evolution. Together with the $y_3$ data we would have 6
parameters, reduced to 5 by the constraint \eqref{consy}, which
turns into $\g_3=-\left(\tfrac34
\g_1^2+2\g_1\g_2+\g_2^2\right)/(\g_1+\g_2)$.

Equation \eqref{y1as} describes also the behavior of $\ln E_k$, as
the contributions of $r$ and $s$ to $E_k$ are negligible in the
right neighborhood of $\ta=0$, while the contributions of $q$ and
$s$ may be neglected in its left neighborhood. The result is
sawtooth shape of the dependence $E_k(\ta)$, see Fig.
\ref{kin-en}. In the graph of $q,\,r,\,s$ (Fig. \ref{qrq}), the
kinetic energy would approximately be their upper envelope, as
$E_k\approx \max({q,r,s})$.

For larger $\ta$-distances from $\ta_+$, if only $k_q>-\beta$,
i.e., $\g_1>-\tfrac23\g_2$, then the increasing function $r$ of
\eqref{sol11} grows fast with $\ta$ until it overtakes the
decreasing function $q$, while the growth of $s$ remains slow, if
any (this depends on $k_s$). Thus $y_1$ is decreasing while $y_2$
is increasing. After some time $\ta$, overtaking of $y_1$ by $y_2$
(and $q$ by $r$) must happen. Later, the reciprocal situation
would affect $r$. As long as $s\ll q$ and $s\ll r$, a sequence of
exchanges between these two variables takes place (see Fig.
\ref{qrq}). This way, we obtain sawtooth oscillations, alternating
as in Fig.~\ref{qrq}.
\begin{figure}
\begin{center}
\includegraphics[width=0.8\textwidth]{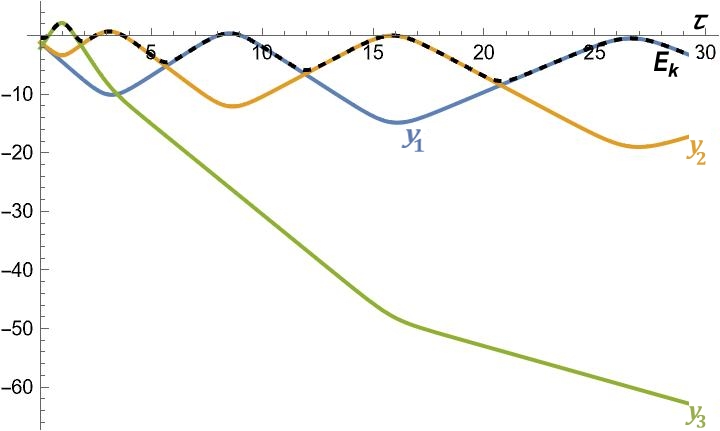}
\caption{The logarithmic variables $y_1=\ln q,\,y_2=\ln r$ and
$y_3=\ln s$ as functions of time parameter $\ta$ for $q(0)=0.3,\,
r(0)=0.33,\,s(0)=0.12,\,\dot{q}(0)=-1,\,\dot{r}(0)=-1$. Variables
$y_1$ and $y_2$ exhibit sawtooth behavior, exchanging their
dominant roles, while $y_3$ is much smaller. The kinetic energy
$E_k$ (dashed line) is the upper envelope of the $y_i$ variables.
The tips of the saw teeth are more rounded than those in the
previous graph (Fig. \ref{kin-en}) as the initial ``velocities''
are less than halves of those in Fig. \ref{kin-en}. Still the term
``sawtooth'' is justified by the rectilinear shape of the teeth
outside their blades.\label{qrq}}
\end{center}
\end{figure}

The above reversal of the ordering $q=a^2$ with $r=b/a$ is
equivalent to reversing the order of $a^3$ and $b$, while the two
are equal (up to a constant 9/10) for the exact solution. Similar
situations (with the exact solution in the middle) occur for the
other reversals.



If we wanted to follow the dynamics beyond the validity of the
above approximation, we might take another step, considering the
system \eqref{L1qrs}, with only one of the exponential functions,
$\exp y_1,\,\exp y_2,\,\exp y_3$ set to $0$ in \eqref{L1qrs},
while two of them remain nonzero. Such systems have been solved
and discussed by Conte \cite{Conte} as systems linked to the BKL
scenario. These systems prove to be useful as approximations of
the actual behavior in the scenario.

With the exception of case $b/a=r=\exp(y_2)\approx 0$, the
solutions are much more complex than those for only a single
nonzero $\exp(y_i)$. For this simplest case, the solution (mutatis
mutandis equivalent to that of \cite{Conte}) reads
\bs\label{sol22}
\begin{align}
q =&k_q^2 \sech^2 k_q (\ta - \ta_-),\label{sol22a}\\
r =&k_r^2 e^{\al \ta} \cosh^2[k_q (\ta - \ta_-)] \cosh[k_s (\ta
- \ta_3)],\label{sol22b}\\
s=&k_s^2 \sech^2[k_s (\ta - \ta_3)],\label{sol22c}
\end{align}
\es
where the constants $\ta_-$ and $\ta_3$ are time values
corresponding to the turning points of $y_1$ and $y_3$
respectively, while $\al,\,k_q,\,k_r$ and $k_s$ have the similar
sense as in the previous three cases. The kinetic energy is $q+s$,
hence it is the sum of the two $\sech^2$ functions.

In \eqref{sol22}, it is evident, that also in this case the
exchange takes place between the initially greater of $q$ and $s$
being surpassed by $r$, even though $r$ was assumed to be $0$ in
the neighborhood of $\ta=\ta_+$. After this exchange, the solution
\eqref{sol22} ceases to be valid.

The other two initial orderings which allow for approximating
$s=0$ or $q=0$ on the r.h.s. of \eqref{L1qrs} lead to more refined
solutions \cite{Conte}, namely to a solution expressed in terms of
the Painlev\'e III function (for $s=0$) or a more complex
combination of exponential functions (for $q=0$). As long as the
neglected third component is actually unnecessary, as in
Fig.~\ref{qrq}, the Conte solutions precisely describe the actual
behavior of the universe near the cosmological singularity within
the BKL scenario.

\section{Conclusions}
In the present work, the BKL scenario for the homogeneous universe
in the neighborhood of the cosmological singularity is
supplemented by a few exact results following from the BKL
asymptotics of the Einstein equations. The Lagrangian-Hamiltonian
version of the dynamics described by those equations is
illustrated and analyzed by means of a simple geometric tool, a
cone of the kinetic energy.

It was known that the approach to the singularity had oscillatory
character \cite{BKL,Ryan}. In this paper, a detailed description
of the oscillations in the neighborhood of their turning points is
derived as a solution to reduced equations. The oscillations are
found to have sawtooth shape in the logarithmic variables and so
is the shape of the logarithm of the kinetic energy. Solutions of
other reduced equations, worked out in \cite{Conte}, may describe
dynamics for a longer time, i.e., as long as one of the variables,
$a^2,\, b/a$, or $c/b$, remains significantly smaller than the
other two as in Fig. \ref{qrq}. Whatever the oscillations be, the
singularity relying on the collapse in all 3 dimensions must
eventually happen in the limit $\ta\to\infty$ (i.e. $t\to 0^+$)
for all initial (or final) conditions corresponding to decreasing
volume $dV/d\ta|_{\ta=0}<0$. Within the scenario, singularities of
exact Kasner's shape do not happen, even in the limit
$\ta\to\infty$. This means that never does the universe shrink to
zero in some directions while extending to infinity in the others,
like in the exact Kasner solutions \eqref{Kasner}, although the
oscillations might approach such solutions closely. No
singularity, of any type, may happen for finite $\ta$. Finally,
the exact solution derived in \cite{GP} is found to be the only
one which describes the approach to the singularity with
well-defined proportions between the directional scale factors (at
the appropriate powers). Since \cite{GP} predicts instability of
that exact solution, this means that the chaos on the approach to
the singularity is inevitable.

\begin{appendices}

\section{Proof that asymptotic of the exact solution is the only path to the
collapse with well-defined proportions between $q,\,r,\,s$
(Proposition \ref{uni-apex})}
\begin{proof} \textit{Plan:}
 We will first (i) prove that all ratios $y_i/y_j,~\,i,j=1,2,3$, tend to
$1$. Then (ii) these limits yield the limits of ratios: $r/q\to
10/9$ and $s/q\to 4/9$, which are identical with the corresponding
limits in the exact solution \eqref{solqrs}. Finally (iii), we
find the asymptotic time dependence of $q,\,r,\,s$ to be that of
the exact solution, i.e., proportional to $(\ta-\ta_0)^{-2}$ with
coefficients $9,\,10,\,4$ respectively.

In these calculations we make use of de l'H\^opital's rule applied
to quotients $r/q$ and $s/q$ (see Remark after Proposition
\ref{uni-apex}) and to ratios $y_i/y_j$ and
$\dot{y}_i/\dot{y}_j,~i,j=1,2,3$. In particular, we have
$\lim_{\ta\to\infty}y_i/y_j=
\lim_{\ta\to\infty}\dot{y_i}/\dot{y_j}=\lim_{\ta\to\infty}\ddot{y_i}/\ddot{y_j}$,
because for all $i,~y_i\to -\infty$, while all $\dot{y}_i\to 0$,
at the apex.

i) We have, from de l'H\^opital's rule
\be
\lim_{\ta\to\infty}\frac{y_2}{y_1}=\lim_{\ta\to\infty}\frac{\ln
r}{\ln q}=\lim_{\ta\to\infty}\frac{\dot{r}/r}{\dot{q}/q}=1.
\ee
For other pairs chosen from $y_1,\,y_2,\,y_3$ and their respective
counterparts $q,\,r,\,s$ the procedure is identical.

ii) We may apply de l'H\^opital's rule twice, as the all-direction
collapse means that the trajectory in the cone tends to the apex,
whence $\dot{y}_i\to 0$ for all $i$. Using also the dynamics
equations \eqref{L1qrs} for replacement of the second derivatives
of $y_i$, we obtain
\be\label{q-limit321}
1=\lim_{\ta\to\infty}\frac{y_2}{y_1}=\lim_{\ta\to\infty}\frac{\ddot{y}_2}{\ddot{y}_1}=\lim_{\ta\to\infty}\frac{2(\exp
y_1-\exp y_2)+\exp y_3}{-2(\exp y_1-\exp
y_2)}=\lim_{\ta\to\infty}\frac{s}{2(r-q)}-1,
\ee
where we have replaced $\exp y_i$ by the appropriate of $q,\,r$,
or $s$. From \eqref{q-limit321}, we get
\be\label{sr-q}
\lim_{\ta\to\infty}s/(r-q)=4
\ee
By similar operations on $y_3/y_1$, we obtain
\be\label{q-limit21}
1=\lim_{\ta\to\infty}\frac{r-2s}{2(r-q)}.
\ee
Substituting \eqref{sr-q} to \eqref{q-limit21}, we obtain
\be\label{r-q}
\lim_{\ta\to\infty} r/(r-q)=10, \text{  whence  }
\lim_{\ta\to\infty} r/q=10/9.
\ee
If the limit of $r/(r-q)$ is substituted from \eqref{r-q} to
\eqref{sr-q}, we obtain
\be\label{sq}
\lim_{\ta\to\infty}(s/r)=4/10, \text{  whence
}\lim_{\ta\to\infty}(s/q)=4/9.
\ee
iii) The asymptotic $\ta$ dependence may be recovered from
\eqref{length}, which in terms of $y_i$ has the form
\be\label{lengthy}
\frac34\dot{y}_1^2+2\dot{y}_1\dot{y}_2+\dot{y}_2^2+\dot{y}_1\dot{y}_3+\dot{y}_2\dot{y}_3
-\frac92\ddot{y}_1-5\ddot{y}_2-2\ddot{y}_3=0.
\ee
Dividing both sides of \eqref{lengthy} by $\dot{y}_1^2$ and making
use of the fact that all quotients $\dot{y}_i/\dot{y}_1$ and
$\ddot{y_i}/\ddot{y_1}$ tend to 1, we obtain the asymptotic, which
in the lowest order may be written as
\be
\lim_{\ta\to\infty}\frac{d}{d\ta}\left(\frac{1}{\dot{y}_1}\right)=-\frac12,
\ee
Integrating, we get the asymptotic of $y_1$ in the neighborhood of
$\ta=\infty$
\be\label{y1}
\dot{y}_1= -2/(\ta-\ta_0),\qquad y_1= \ln \frac{C}{(\ta-\ta_0)^2}.
\ee
While the value of $\ta_0$ is arbitrary, the value of $C$ may be
recovered by substitution of \eqref{y1} into the constraint
\eqref{consy}, which yields $C=9$. Substitution of this $C$ into
\eqref{y1} yields $y_1$ and $q=\exp y_1$ as in \eqref{solqrs},
then the asymptotics of $r$ and $s$ can be obtained from
\eqref{r-q} and \eqref{sq} respectively. This way, the asymptotic
of any solution with well-defined limits of $r/q$ and $s/q$ proves
to be identical with the exact solution \eqref{solqrs}, equivalent
to \eqref{sol-u} and \eqref{solution} (for $\ta>\ta_0$).
\end{proof}
\section{Proof that all trajectories eventually tend to the apex
(Proposition \ref{g_i=0})}
We will prove it in two stages, using
the $y_i,~i=1,2,3$ variables of \eqref{qrs2y} (which are connected
with the scale factors by \eqref{a2qrs} and with the $u_i$ through
\eqref{y2u}).

\underline{Stage 1}
\begin{lemma}
 Let $\g_1, \g_2, \g_3$ be
the limits of $\dot{y}_1,\dot{y}_2,\dot{y}_3$ (respectively) at
$\ta\to\infty$. Then $\g_1=\g_2=0$, while $\g_3\le 0$.
\end{lemma}
\begin{proof}
If for $\ta\to\infty$, the trajectory approaches the surface of
the cone, then the kinetic part in the constraint \eqref{consy}
turns to zero. In terms of the limits $\g_i$
\be\label{Eky=0}
\frac34 \g _ 1^2 +
 2 \g_1 \g _ 2 + \g _ 2^2 + \g _ 2 \g_3 + \g_3
 \g_1=0
\ee
The constraint \eqref{consy} requires that the potential part also
turns to 0. This means that the asymptotics $y_i= \g_i
t+o(t),~i=1,2,3$ has all $\g_i\le 0$. To also satisfy
\eqref{Eky=0}, the first two of the $\g_i$'s must be equal to
zero.
\end{proof}

\underline{Stage 2}: proof of Proposition \ref{g_i=0}
\begin{proof}
In the Lemma, we proved $\g_1=\g_2=0$, $\g_3\le 0$. We are going
to prove that $\g_3<0$ is impossible.

Assume $\g_3<0$. Adding first two equations in the right equation
of \eqref{L1qrs}, we obtain
\be\label{y12}
\ddot{y}_1+\ddot{y}_2=\exp{y_3}.
\ee
Since  $\lim_{\ta\to\infty}\dot{y}_3=\g_3<0$, then for all $\ep>0$
a time parameter $T$ exists such that for all $\ta>T$, we have
\be
\dot{y}_3\in\, ]\g_3-\ep,~\g_3+\ep[,~ \text{ whence
}~y_3-y_3(T)\in \,](\g_3-\ep)(\ta-T),~(\g_3+\ep)(\ta-T)[\,.
\ee
Choose $\ep$ such that $\g_3+\ep<0$. We have, from \eqref{y12},
\be
\ddot{y}_1+\ddot{y}_2\in \left.
\right]e^{y_3(T)+(\g_3-\ep)(\ta-T)},~e^{y_3(T)+(\g_3+\ep)(\ta-T)}\left[
\right.\,,
\ee
with both exponents negative for sufficiently large $\ta$. Hence,
for these $\ta$,
\be\label{dot-ineq}
\dot{y}_1+\dot{y}_2\in \left.
\right]\frac{1}{\g_3-\ep}e^{y_3(T)+(\g_3-\ep)(\ta-T)}+C_1,~\frac{1}{\g_3+\ep}e^{y_3(T)+(\g_3+\ep)(\ta-T)}+C_1\left[
\right.~,
\ee
where $C_1$ is a constant of integration. Since
$\g_1=\lim_{\ta\to\infty}\dot{y}_1=0$ and
$\g_2=\lim_{\ta\to\infty}\dot{y}_2=0$ (from the Lemma), we have
$C_1=0$. Then, integrating again \eqref{dot-ineq}, we obtain
\be
{y}_1+{y}_2\in \left.
\right]\frac{1}{(\g_3-\ep)^2}e^{y_3(T)+(\g_3-\ep)(\ta-T)}+C_2,~\frac{1}{(\g_3+\ep)^2}e^{y_3(T)+(\g_3+\ep)(\ta-T)}+C_2\left[
\right.~,
\ee
where $C_2$ is a constant of the next integration. However the
constraint \eqref{consy} requires that both ${y}_1\to -\infty$ and
${y}_2\to -\infty$ as $\ta\to\infty$, while both ends of the
interval on the r.h.s. tend to a finite $C_2$. Hence, the
assumption $\g_3<0$ leads to a contradiction, whence $\g_3=0$.

The conclusion that all $\g_i,~~i=1,2,3$ are equal to zero means
that all trajectories tend to the apex.
\end{proof}
\end{appendices}

\end{document}